\documentclass[conference,a4paper]{IEEEtran}

\usepackage{amssymb, amsmath, color}
\usepackage{amsthm,graphicx,enumerate,verbatim,xcolor}
\usepackage[small]{caption}
\usepackage{subfigure}
\usepackage{tikz}

\renewcommand{\baselinestretch}{.97}

\def\squarebox#1{\hbox to #1{\hfill\vbox to #1{\vfill}}}

\newcommand{\inout}         {{- \hspace{-0.3mm} \vbox{\hrule\hbox{\vrule\squarebox{.667em}\vrule}\hrule}\hspace{-0.3mm} -}}

\DeclareMathOperator*{\argmin}{arg\,min}

\newtheorem{thm}{Theorem}

\newtheorem{lem}{Lemma}
\newtheorem{defn}{Definition}

\begin{document}

\sloppy

%% Paper Title
%% You can use linebreaks \\ within to get better formatting as
%% desired. 
\title{A Bit of Secrecy for Gaussian Source Compression}

%% Author names and affiliations:

%% For up to three affiliations:
%%
%\author{
%  \IEEEauthorblockN{Eva C. Song}
%  \IEEEauthorblockA{Dep. of Electrical Eng.\\
%    Princeton University\\
%    Email: csong@princeton.edu} 
%  \and
%  \IEEEauthorblockN{Paul Cuff}
%  \IEEEauthorblockA{Dep. of Electrical Eng.\\
%    Princeton University\\
%    Email: cuff@princeton.edu}
%  \and
%  \IEEEauthorblockN{H. Vincent Poor}
%  \IEEEauthorblockA{Dep. of Electrical Eng.\\ 
%    Princeton University\\
%    Email: poor@princeton.edu}
%}
\author{Eva C. Song \qquad Paul Cuff \qquad H. Vincent  Poor\\ Dept. of Electrical Eng., Princeton University,  NJ 08544\\ \{csong, cuff, poor\}@princeton.edu}

%% To balance the two columns, you should reduce the text-height of
%% the last page using the following command:
%%%%%%%%%%%%%%%%%%%%%%%%%%%%%%%%%%%%%%%%%%%%%%%%%%%%%%%%%%%%%%%%%%%%%
%\addtolength{\textheight}{-9.35cm}
%%%%%%%%%%%%%%%%%%%%%%%%%%%%%%%%%%%%%%%%%%%%%%%%%%%%%%%%%%%%%%%%%%%%%
%% with an appropriate value. This command must be place on the second
%% last page, i.e., for a one-page abstract here, for a two-page
%% abstract right after the \maketitle command.

%% Create the title:
\maketitle

\begin{abstract}
In this paper, the compression of an independent and identically distributed Gaussian source sequence is studied in an unsecure network.  Within a game theoretic setting for a three-party noiseless communication network (sender Alice, legitimate receiver Bob, and eavesdropper Eve), the problem of how to efficiently compress a Gaussian source with limited secret key in order to guarantee that Bob can reconstruct with high fidelity while preventing Eve from estimating an accurate reconstruction is investigated. It is assumed that Alice and Bob share a secret key with limited rate. Three scenarios are studied, in which the eavesdropper ranges from weak to strong in terms of the causal side information she has. It is shown that one bit of secret key per source symbol is enough to achieve perfect secrecy performance in the Gaussian squared error setting, and the information theoretic region is not optimized by joint Gaussian random variables.
%We find an upper bound on the number of key bits demanded per source symbol which forces maximum distortion to Eve (i.e. Eve cannot do better in recontructing the source than pure guessing based merely on the source distribution). When less key bits are available, a Gaussian quantizer is selected to account for the tradeoff between fidelity and secrecy.
\end{abstract}

\section{Introduction}
This work specializes general results in the literature for secure source compression found in \cite{allerton-cuff}, \cite{globecom-cuff} and \cite{schieler-isit} to the important case of a Gaussian source and squared-error distortion. We analytically optimize the information-theoretic secrecy regions in only certain regimes, but in doing so we draw out some interesting observations: No more than 1 bit/symbol of secret key is ever needed for security in the distortion sense; and jointly-Gaussian selection of variables is sub-optimal.

Since the entropy of a Gaussian random variable is infinite, it would require an infinite amount of bits to losslessly represent this variable. Traditional rate distortion theory introduced by Shannon \cite{shannon-rd} opens the door for lossy compression, which characterizes how to effectively compress a source within certain tolerance for distortion in the setting of point-to-point communication. Secrecy in an unsecure network with potential eavesdroppers has been studied in the past. Such work generally uses equivocation as a metric for secrecy, which measures how much uncertainty about the source is reduced given the eavesdropper's observation. Perfect secrecy requires almost no reduction in the uncertainty. Equivocation is an important tool for establishing perfect secrecy in physical layer security. Wyner \cite{wyner} and Csisz\'{a}r and K\"{o}rner \cite{ck} show how secure transmission can be made possible over a noisy broadcast channel without using any secret key, and they also use equivocation as a metric of partial secrecy when perfect secrecy is not achieved. However, when dealing with source coding in noiseless channels (a more common assumption in cryptography), equivocation no longer seems to capture the essence of the problem because it does not tell us how Eve can make sense out of her observation in estimating the source. It is well known, according to Shannon \cite{shannon-secrecy}, that perfect secrecy would require the same amount of key bits as information bits. In particular, for a source with high or infinite entropy, e.g. a Gaussian source, where perfect secrecy is too expensive to obtain, how much useful information must be revealed? Can we leak information to Eve that is less meaningful (e.g. the high-precision bits)?

While perfect secrecy remains a luxury for secure communication, Yamamoto \cite{yamamoto} takes an alternative approach by measuring the secrecy with distortion given that the eavesdropper is forced to reconstruct the source. For example, suppose Alice is a cable television provider and Bob is a customer of Alice that pays a service fee. Eve is a cable thief who does not pay Alice and wants to watch the same programs available to Bob by intercepting the communication. Alice encodes the programs with some secret key she shares with Bob so that Bob can decode and watch the programs. Even when the key is limited, Alice can still choose to encode the important part (e.g. the video) of the programs. Consequently, the signal Eve gets may be correlated with the programs but she cannot watch them even though perfect secrey is compromised. More recent work by Cuff \cite{allerton-cuff} established a general model for looking at secrecy as a maxmin problem, in which the payoff function typically captures two major criteria: reliability (i.e. Bob's reconstruction of the source is close enough to the original) and secrecy (i.e. Eve's reconstruction of the source is far away enough from the original). The special case when lossless compression is required by Bob was studied in \cite{globecom-cuff} in which the payoff simplifies to distortion between Alice and Eve. The interesting case in which the distortion measure is hamming distance and Eve has causal information about the source realization was solved in \cite{allerton-schieler}.

In this paper, we investigate the model proposed in \cite{allerton-cuff} for an independent and identically distributed (i.i.d.) Gaussian source sequence. Lossy compression between Alice and Bob is considered. The causal information available to Eve in estimating the source may differ depending on the application. When the source sequence needs to be decoded all at once, no causal information (past moves of Alice and/or Bob) should be assumed at Eve's decoder. On the other hand, when the source is decoded in sequence, it is more appropriate to assume that Eve gets to see the past moves. The assumption of causal information disclosure makes this distortion notion of secrecy a more robust metric and equivocation becomes a special case of this problem by choosing the distortion function properly \cite{cuff-ita}, \cite{schieler-journal}. Different scenarios regarding the causal information availability are considered which require solving several optimization problems. Unlike the lossless case in \cite{globecom-cuff} where it reduces to a linear programming problem, the constraint region for lossy compression is much more complicated. Several schemes are discussed, some of which achieve optimality under certain conditions. A scalar quantization special case is addressed and numerical result is given for all cases.

%In this work, we study the secrecy of transmitting Gaussian source over a noiseless network in %presence of an eavesdropper, extending the work in \cite{allerton-cuff}. 

\section{Problem Setup}

Let us consider the network in Fig. \ref{network}. The sender Alice has an i.i.d. Gaussian source sequence $X^n$, where each $X_i\sim \mathcal{N}(\mu_0, \sigma_0^2)$, which we denote by $p_0$. Alice compresses $X^n$ at rate $R$ and sends it to the legitimate receiver Bob through a noiseless rate-limited channel. In addition to the channel, Alice and Bob can make use of a shared secret key $K$ during the communication. An eavesdropper Eve intercepts the communication and observes the message $M\in [1:2^{nR}]$, which will have been at least partially encrypted using $K$. Both Bob and Eve are trying to reproduce the source sequence, i.e. at each time $i$, Bob estimates the current symbol $X_i$ based on the encrypted message $M$ and some side information (SI) $S^{B}_i$ while Eve does the same based on her own SI $S^{E}_i$. 

Specifically, we have the secret key
$K\sim Unif[1:2^{nR_s}]$, the encoder
$f_{n}: \mathcal{X}^n \times \mathcal{K} \mapsto \mathcal{M}$, and the decoder for Bob
$\{g_i: \mathcal{M} \times \mathcal{S}^{B}_i \mapsto \mathcal{Y}\}_{i=1}^{n}$. Similarly, Eve has her  decoder $\{t_i: \mathcal{M} \times \mathcal{S}^{E}_i \mapsto \mathcal{Z}\}_{i=1}^{n}$. More generally, $f_n$, $\{g_i\}_{i=1}^n$, and $\{t_i\}_{i=1}^n$ are stochastic functions of $p(m|x^n, k)$, $\{p(y_i|m,s^{B}_i)\}_{i=1}^n$, and $\{p(z_i|m,s^{E}_i)\}_{i=1}^n$, respectively.

\begin{figure}[htbp]
  \centering
  \includegraphics[width=8 cm]{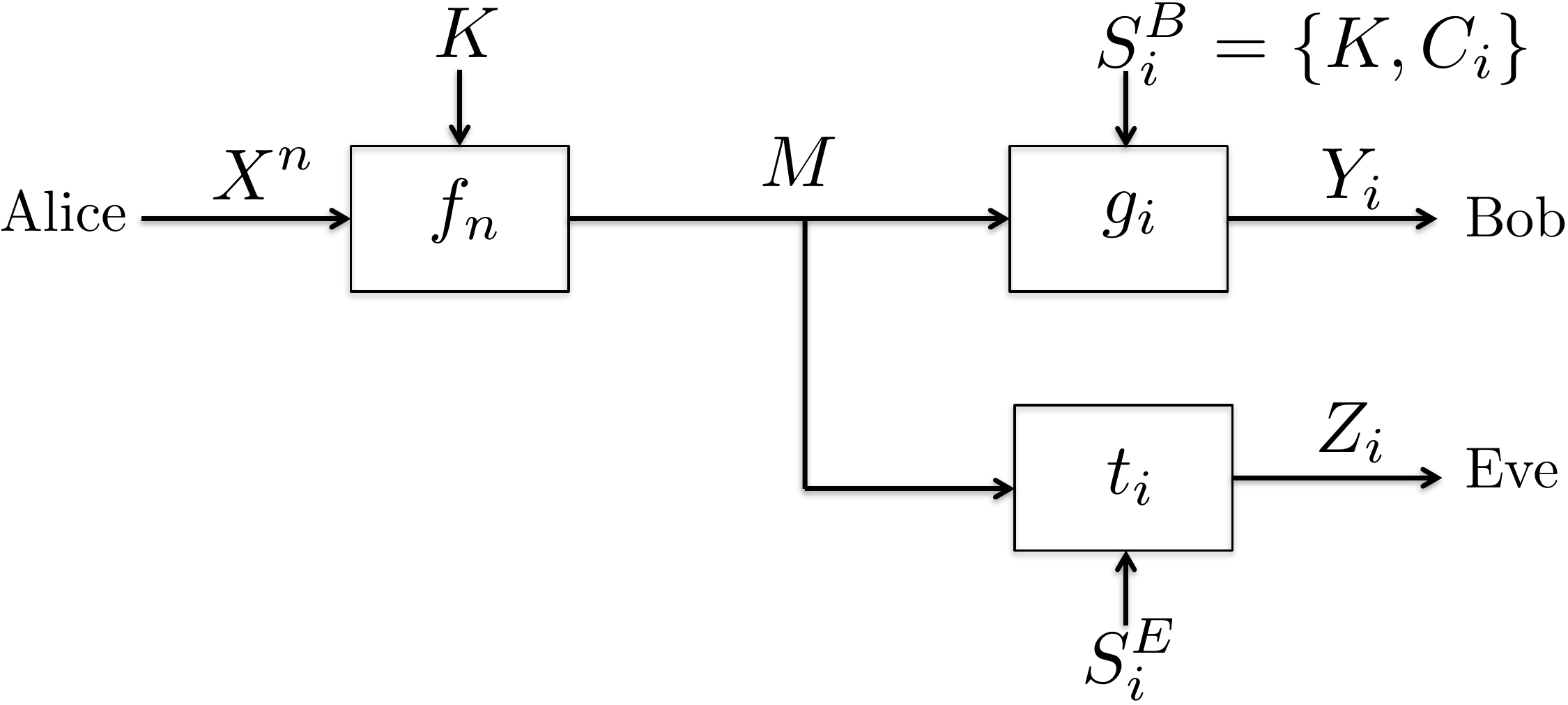}
\caption{The encoder encodes the whole source sequence $X^n$ with secret key $K$. Decoders estimate one symbol at each time instance $i$. }
\label{network}
\end{figure}

To measure the fidelity of Bob's reconstruction and the secrecy of the source with repect to Eve, we define the following payoff function:
\begin{eqnarray}
\pi(x,y,z) &\triangleq& \frac{1}{\sigma_0^2}[(z-x)^2-(y-x)^2]. \label{payoff}
\end{eqnarray}
Note that $(\ref{payoff})$ compares the squared error distortions of Bob and Eve. If the distance between Eve's symbol $z$ and the original source symbol $x$ is greater than that of Bob's symbol $y$ and $x$, we have a positive payoff; otherwise, we get a negative payoff. The payoff function is normalized so that the result does not depend on the variance of the source. The payoff of sequences is defined as the average of per letter payoff:
$$\pi(x^n,y^n,z^n)\triangleq \frac{1}{n}\sum_{i=1}^n \pi(x_i,y_i,z_i).$$
\begin{defn}
A secrecy rate-payoff triple $(R, R_s, \Pi)$ is achievable if $K\in[1:2^{nR_s}]$, $M\in[1:2^{nR}]$, and 
$$\lim_{n\rightarrow \infty}\sup_{\{f_n, \{g_i\}_{i=1}^n\}}\inf_{\{t_i\}_{i=1}^n}\mathbb{E}\pi(X^n, Y^n, Z^n)\geq\Pi.$$
\end{defn}

In this paper, we will consider the following three scenarios of SI and evaluate the corresponding secrecy rate-payoff regions. 
\subsection{Weak Eavesdropper}
%\vspace{-.7cm}
\begin{eqnarray*}
%s^{B}_i&=&\{k\}\\
s^{E}_i&=&\{z^{i-1}\}
\end{eqnarray*}
\subsection{Causal Source Awareness}
%\vspace{-.7cm}
\begin{eqnarray*}
%s^{B}_i&=&\{k\}\\
s^{E}_i&=&\{x^{i-1}, z^{i-1}\}
\end{eqnarray*}
This corresponds to the scenario in which a stronger eavesdropper gets to see the past realization of the source symbols when estimating the current symbol.
\subsection{Causal General Awareness}
%\vspace{-.7cm}
\begin{eqnarray*}
%s^{B}_i&=&\{k, x^{i-1}, y^{i-1}, z^{i-1}\}\\
s^{E}_i&=&\{x^{i-1}, y^{i-1}, z^{i-1}\}
\end{eqnarray*}
Here all causal information is available to Eve for reconstructing the Gaussian source sequence.

In all three cases, the SI of Bob is given by $k$ and $c_i=\{x^{i-1}, y^{i-1}, z^{i-1}\}$. It was shown in \cite{allerton-cuff} that it is sufficient to consider \textbf{only} $s^{B}_i=\{k\}$ which is independent of Bob's awareness of the eavesdropper. 

\section{Secrecy Rate-Payoff Regions}
\begin{thm}
The secrecy rate-payoff triple $(R, R_s, \Pi)$ for a \textbf{weak eavesdropper} is achievable for an i.i.d. Gaussian source if and only if
$$R_s>0$$
and
$$\Pi \leq 1-\exp(-2R).$$
\end{thm}
This is an application of \cite{schieler-isit}, in which it is shown that for lossless compression between Alice and Bob, any strictly positive key rate will guarantee the maximum average distortion between Alice and Eve. We then apply the distortion-rate function for Gaussian source $d(R)=\sigma_0^2 \exp(-2R)$ to obtain this result. Note that because of the way we define the payoff function in $(\ref{payoff})$, the result does not depend on the mean or the variance of the source. 

Notice that the payoff achieved for any positive $R_s$ is as good as what would be achieved under perfect secrecy.

With \textbf{causal source awareness}, the general form of the optimal payoff function was characterized in \cite{allerton-cuff} as 
\begin{eqnarray}
\Pi_{p_0}(R, R_s)&=&\max_{p(y,u|x)\in\mathcal{P}} \min_{z(u)} \mathbb{E} \pi(X,Y,z(U)) \label{paycurve} \\
\mathcal{P}&=&\left\{ 
     \begin{array}{lr}
       p(y,u|x):\\
       R_s\geq I(X;Y|U)\\
       R\geq I(X;U,Y)
     \end{array}
   \right\}, \nonumber
\end{eqnarray}
where $\Pi_{p_0}(R,R_s)$ is the maximum $\Pi$ that is achievable with rates $R$ and $R_s$. A proof of this payoff function can be found in \cite{schieler-journal}. To specialize this result to the Gaussian case, we must optimize over the choice of distribution $p(y,u|x)$. It is often the case for Gaussian problems that the choice of jointly Gaussian auxiliary random variables is optimal. However, because Bob and Eve are playing the game competitively instead of collaboratively, it turns out that a jointly Gaussian distribution does as poorly as if no auxiliary random variable is used.  Below we analyze the best payoff that can be achieved with jointly Gaussian distributions and later show that other non-Gaussian choices of $p(y,u|x)$ can do better.

First, we observe the following from $(\ref{payoff})$ and $(\ref{paycurve})$:
\begin{eqnarray}
&&\Pi_{p_0}(R, R_s)\nonumber\\
&=&\max_{p(y,u|x)\in\mathcal{P}} \min_{z(u)} \frac{1}{\sigma_0^2}\mathbb{E}[(z(U)-X)^2-(Y-X)^2] \nonumber\\
&=& \frac{1}{\sigma_0^2}\max_{p(y,u|x)\in\mathcal{P}}[\sum_{x,u}p(u|x)p_0(x)(x-\mathbb{E}[X|U=u])^2 \nonumber\\
&\ \ &-\sum_{x,y}p(y|x)p_0(x)(y-x)^2]  \label{mmse}
\end{eqnarray}
where $(\ref{mmse})$ comes from the fact that the conditional mean is the minimum mean squared error (MMSE) estimator. Here we are using $\sum$ and $\int$ interchangebly for convenience because it is not clear whether $Y$ and $U$ are discrete or not at this stage.

\subsection{Jointly Gaussian}
\begin{thm} \label{causal}
The solution to $(\ref{paycurve})$ for a Gaussian source distribution $p_0$ is
$$\Pi_{p_0}(R,R_s)=1-\exp(-2\min(R_s,R))$$
when $p(x,y,u)$ is constrained to be a jointly Gaussian distribution.
\end{thm}

\begin{proof} 
Achieving this payoff is easy. We only need to choose $p(y|x)$ such that $X$ and $Y$ are jointly Gaussian and $I(X;Y)\leq \min(R_s, R)$. We then choose $U$ to be Gaussian and independent of $X$ and $Y$. 

To show the converse, we first fix $p(x,y,u)$ to be jointly Gaussian and write $(\ref{paycurve})$ and $(\ref{mmse})$ in terms of the correlations $\rho_{xy}$, $\rho_{xu}$, and $\rho_{yu}$. It can be verified that maximizing the objective funtion in $(\ref{paycurve})$ is equivalent to maximizing $g(\rho_{xy}, \rho_{xu}, \rho_{yu})\triangleq \rho_{xy}^2-\rho_{xu}^2$. Ignoring the constraints on valid choice of $\rho_{xy}$, $\rho_{xu}$, $\rho_{yu}$ that are implied by any covariance matrix, we have two constraints from $(\ref{paycurve})$:
\begin{eqnarray}
%0&\leq&1-\rho_{xy}^2-\rho_{xu}^2-\rho_{yu}^2+2\rho_{xy}\rho_{xu}\rho_{yu}\\
 R_s&\geq& \frac{1}{2}\log \frac{(1-\rho_{xu}^2)(1-\rho_{yu}^2)}{1-\rho_{xy}^2-\rho_{xu}^2-\rho_{yu}^2+2\rho_{xy}\rho_{xu}\rho_{yu}}\label{r0constraint}\\
 R&\geq& \frac{1}{2}\log \frac{(1-\rho_{yu}^2)}{1-\rho_{xy}^2-\rho_{xu}^2-\rho_{yu}^2+2\rho_{xy}\rho_{xu}\rho_{yu}} \label{rconstraint}.
\end{eqnarray}
%By removing constraint $(\ref{rconstraint})$ and looking at $(\ref{r0constraint})$ by itself, we have $g(\rho_{xy}, \rho_{xu}, \rho_{yu})\leq 1-\exp(-2R_s) + \delta(\rho_{xy},\rho_{xu},\rho_{yu})$, where $\delta(\rho_{xy},\rho_{xu},\rho_{yu})=\exp(-2R_s)(\rho_{xu}^2+\rho_{yu}^2-\rho_{xu}^2\rho_{yu}^2)-2\rho_{xu}^2-\rho_{yu}^2+2\rho_{xy}\rho_{xu}\rho_{yu}$. Without loss of generality, we can assume $\rho_{xu}\rho_{yu}\geq 0$. If $\rho_{xy}<0$, $\delta(\rho_{xy},\rho_{xu},\rho_{yu})<0$ trivially. Otherwise, $\delta(\rho_{xy},\rho_{xu},\rho_{yu})\leq \delta(R_s)$, where $\delta(R_s)=\exp(-2R_s)(\rho_{xu}^2+\rho_{yu}^2-\rho_{xu}^2\rho_{yu}^2)-2\rho_{xu}^2-\rho_{yu}^2+2\rho_{xu}\rho_{yu}$. Since $\delta(R_s)$ monotonically decreases with $R_s$ and $\delta(Rs)|_{Rs=0}\leq 0$, $\delta(R_s)\leq 0$. Therefore, $g(\rho_{xy}, \rho_{xu}, \rho_{yu})\leq 1-\exp(-2R_s)$. Similarly, by removing constraint $(\ref{r0constraint})$ and keeping $(\ref{rconstraint})$, we would obtain $g(\rho_{xy}, \rho_{xu}, \rho_{yu})\leq 1-\exp(-2R)$. Hence, the objective function under all constraints needs to satisfy: $g(\rho_{xy}, \rho_{xu}, \rho_{yu})\leq 1-\exp(-2\min(R_s, R))$.
By removing constraint $(\ref{rconstraint})$, with some algebraic manipulation, it can be shown that $g(\rho_{xy}, \rho_{xu}, \rho_{yu})\leq 1-\exp(-2R_0)$. Similarly, by removing constriant $(\ref{r0constraint})$, it can be shown that $g(\rho_{xy}, \rho_{xu}, \rho_{yu})\leq 1-\exp(-2R)$. Hence, the objective function under all constraints needs to satisfy $g(\rho_{xy}, \rho_{xu}, \rho_{yu})\leq 1-\exp(-2\min(R_s, R))$.
\end{proof}
Theorem \ref{causal} implies that for a jointly Gaussian distribution $p(x,y,u)$, choosing the auxiliary random variable $U$ correlated with $X$ and $Y$ does not improve the payoff from an uncorrelated $U$. In this case, $U$ does not give out any information about $X$ and therefore the distortion between Alice and Eve is kept to a maximum as if under perfect secrecy. However, the rate-distortion tradeoff between Alice and Bob is limited by the secret key rate. Is it possible to achieve a higher payoff by another choice of $p(y,u|x)$? Below we show how a simple Gaussian quantization can provide a better solution that is independent of the key rate $R_s$ under certain conditions. 

\subsection{Gaussian Quantization} \label{sec-quan}
Let us consider the following construction. We first fix $X\sim p_0$. $X$ is quantized symmetrically about its mean with uniform intervals $T$ as shown in Fig. \ref{quantize} so that $Y\triangleq nT, n\triangleq\argmin_{k\in \mathbb{Z}}|kT-X|$, and $U=|Y|$. With this construction, $Y$ is a function of $X$ and $U$ is a function of $Y$. We denote the Markov relationship by $X\inout Y\inout U$. The reason for choosing such a symmetric quantization is that, to maintain a high distortion between Alice and Eve in $(\ref{mmse})$, we want to keep $\mathbb{E}[X|U=u]$ unbiased for all $u$. 

Then the two constraints in $(\ref{paycurve})$ become
\begin{eqnarray}
R_s&\geq& I(X;Y|U)=H(Y|U)=s \text{ bits}, s< 1\label{r_s}\\
R&\geq& I(X;U,Y)=H(Y). \label{r}
\end{eqnarray}

\begin{figure}[htbp]
  \centering
  \includegraphics[width=8 cm]{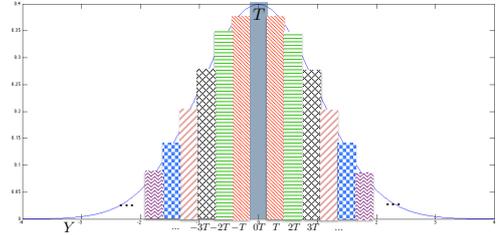}
\caption{Symmetric quantization of the Gaussian random variable $X$ with uniform interval $T$. }
\label{quantize}
\end{figure}

Here we apply the operational meaning of differential entropy from Theorem 9.3.1 of \cite{cover} to get 
\begin{eqnarray*}
H(Y)+\log{T} &\rightarrow& h(X),  \text{ as }T\rightarrow 0,
\end{eqnarray*}
where $h(X)$ denotes the differential entropy of $X$. Recall that the differential entropy of a Gaussian random variable is $h(X)=\frac{1}{2}\log(2\pi e\sigma_0^2)$. Therefore, for $R_s\geq1 \text{ bit}$, as $T\rightarrow 0$, a sufficient condition for $(\ref{r})$ is $T\geq\sqrt{2\pi e}\sigma_0\exp(-R)$. The distortion between Alice and Eve under this Gaussian quantization scheme can be asymptotically calculated as follows:
\begin{eqnarray*}
D^{\Delta}(R)&=&\int_{-\infty}^\infty\sum_y p(y|x)\frac{1}{\sqrt{2\pi}\sigma_0}e^{-\frac{x^2}{2\sigma_0^2}}(y-x)^2 dx\\
&=&\int_{-\infty}^\infty\frac{1}{\sqrt{2\pi}\sigma_0}e^{-\frac{x^2}{2\sigma_0^2}}(nT-x)^2dx\\
&\leq&\frac{\pi e}{2}\sigma_0^2 \exp(-2R)
\end{eqnarray*}
Summarizing the above analysis, we have 
$$\Pi_{p_0}(R,R_s)\geq 1-\frac{\pi e}{2}2^{-2R} \text{ for } R_s\geq 1\text{ bit} \text{ and } R\rightarrow \infty.$$
With Gaussian quantization, even though we sacrifice a constant factor on the distortion between Alice and Bob, the overall payoff is no longer governed by the key rate $R_s$ given that $R_s\geq1  \text{ bit}$. This illustrates that the jointly Gaussian distribution does not achieve optimal payoff. In Section \ref{simulation}, another choice of $U$ with the same quantizer will be discussed for arbitary $R$ and $R_s$ that can perform better than jointly Gaussian for low key rate $R_s$. Next, we will provide a construction for $Y$ and $U$ that achieves maximum payoff under certain conditions.

\subsection{Optimal Payoff for $R_s\geq 1$ bit}
\begin{thm} \label{thm-causal}
If the key rate $R_s\geq 1 \text{ bit}$, the optimal secrecy rate-payoff function for an i.i.d. Gaussian source and \textbf{causal source awareness} is given by
$$\Pi_{p_0}(R,R_s)=1-2^{-2R}.$$
\end{thm}

\begin{proof}
To see the converse, observe that by relaxing the constraints in $(\ref{paycurve})$ and evaluating $(\ref{mmse})$, we have
\begin{eqnarray}
\max_{p(y,u|x)}\sum_{u,x}p(u|x)p_0(x)(x-\mathbb{E}[X|U=u])^2=\sigma_0^2&&\label{distortion_eve}\\
\min_{p(y,u|x): I(X;Y)\leq R}\sum_{x,y}p(y|x)p_0(x)(y-x)^2=\sigma_0^2 2^{-2R}&& \label{distortion_bob}
\end{eqnarray}
%$(\ref{distortion_eve})$ comes from law of total variance, i.e. $\mathbb{E}_U[Var(X|U)]\leq Var(X)$, and 
where $(\ref{distortion_bob})$ comes from the distortion rate function $d(R)$ of Gaussian source. To show the achievability, we choose $p(y,u|x)$ as follows. $Y$ is chosen such that $X$ and $Y$ are zero-mean jointly Gaussian. $U\triangleq|Y|$ and $V\triangleq sgn(Y)$, where $sgn(Y)$ is a binary variable indicating the sign of $Y$. Observe that $U,V$ together gives $Y$. By this construction, we have 
\begin{eqnarray}
I(X;Y,U)&=&I(X;Y)\nonumber\\
I(X;Y|U)&=&I(X;V|U) \ \ <\ \ 1 \text{ bit}.\label{ineq}
\end{eqnarray}
Therefore, given that $R_s\geq 1 \text{ bit}$, the constraint $R_s\geq I(X;Y|U)$ is automatically satisfied from $(\ref{ineq})$. In addition, $\mathbb{E}[X|U=u]=\frac12\mathbb{E}[X|Y=u]+\frac12\mathbb{E}[X|Y=-u]=0$ for all $u$ due to symmetry. The payoff achieves $1-2^{-2R}$.
\end{proof}

The optimization problem $(\ref{paycurve})$ for the case $R_s<1\text{ bit}$ involves the coupling of two terms in $(\ref{mmse})$. In the next section, we present a special case for this regime using a  Gaussian quantizer. 

For the causal general awareness where a strong eavesdropper has the causal information from all parties (Alice, Bob and Eve), the secrecy rate-payoff region was studied in \cite{allerton-cuff} and the result under a general payoff function and i.i.d. source distribution was given in Theorem 4.1 of \cite{allerton-cuff} as the following:
\begin{eqnarray}
\Pi(R,R_s)&=&  \max_{p(y,u,v|x)\in\mathcal{P}}\min_{z(u)} \mathbb{E}\pi(X,Y,z(U)) \label{payoff-gen}\\
\mathcal{P}&=&\left\{ 
     \begin{array}{lr}
       p(y,u,v|x):\\
       p(y|u,v,x)=p(y|u,v)\\
       R_s\geq I(X,Y;V|U)\\
       R\geq I(X;U,V)
     \end{array}
   \right\}. \nonumber
\end{eqnarray}
It turns out that Theorem \ref{thm-causal} also holds for the causal general awareness case. Because of the Markov relation $X\inout (U,V) \inout Y$, we have $R\geq I(X;U,V)\geq I(X;Y)$, which gives us the same result as in Theorem \ref{thm-causal}. With the same choice of $Y$, $U$, and $V$, $I(X,Y;V|U)=1 \text{ bit}$ constantly.
%Because of the markov relation $X\inout (U,V) \inout Y$, we have $R\geq I(X;U,V)\geq I(X;Y)$, which gives us the upper bound $\Pi(R,R_s)\leq 1-2^{-2R}$ in the same way as in the proof of Theorem \ref{thm-causal}. To achieve this payoff for $R_s\geq 1 \text{ bit}$, we do the same construction as in Theorem \ref{thm-causal} and therefore $I(X,Y;V|U)=1 \text{ bit}$ constantly.
Concluding this section,  we have shown that $R_s=1 \text{ bit}$ is the maximum key rate needed even in the case of a strong eavesdropper who is aware of the past moves of all parties.

\section{Gaussian Quantizer Special Case}
In this section, we will address a special structure of the system for the case of causal source awareness. A symbol-by-symbol quantization of the source sequence is performed before the transmission and the legitimate receiver Bob in this system happens to reproduce the scalar quantization of the source at the decoder. That is, we are restricting to a subset $\mathcal{B}_n$ of all valid encoder and decoder pairs $(f_n$, $\{g_i\}_{i=1}^n)$.

Let $\hat{X}_i\sim\hat{p}_0$ be the conditional mean of a uniform quantization of $X_i$ as in Section \ref{sec-quan}, i.e. $\hat{X}_i=\mathbb{E}[X_i|\text{Quantization bin of } X_i]$, and $\hat{M}$ be the encoded message to be transmitted. Let the optimal payoff function under such restriction be $\Pi_{p_0}^{\Delta}(R,R_s)$. The following lemma indicates that revealing the causal realization of the original source is equivalent to revealing the causal realization of the quantized version of the source in Eve's estimate. 
\begin{lem} \label{lemma1}
${X}_i\inout (\hat{M},\hat{X}^{i-1})\inout X^{i-1}$ for all $i=1,...,n$.
\end{lem}
This can be verified by observing that $X^i\inout \hat{X}^i\inout \hat{M}$. Therefore, we can alternatively analyze the performance of the system in Fig. \ref{alternative}. This model is defined formally as the following.
\begin{defn} \label{definition2}
The rate-distortion triple $(R,R_s, D)$ is achievable if 
$$\mathbb{P}[\hat{Y}^n\neq\hat{X}^n]\rightarrow 0 \text{ as } n\rightarrow \infty, \text{ and}$$
$$\lim_{n\rightarrow \infty}\sup_{\{\hat{f}_n, \{\hat{g}_i\}_{i=1}^n\}}\inf_{\{\hat{t}_i(\hat{m},\hat{s}^{E}_i)\}_{i=1}^n}\mathbb{E}\bigg{[}\frac1n \sum_{i=1}^n (\hat{Z}_i-\hat{X}_i)^2\bigg{]}\geq D.$$
\end{defn}

\begin{lem} \label{lemma2}
$X_i\inout\hat{X}_i\inout (\hat{M},\hat{X}^{i-1})$ for all $i=1,...,n$.
\end{lem}

Now by applying Theorem 4.1 from \cite{globecom-cuff}, we have that $(R,R_s,D)$ is achievable iff  $D\leq D_{\hat{p}_0}(R,R_s)\triangleq\max_{p(\hat{u}|\hat{x})\in\mathcal{Q}}\min_{\hat{z}(\hat{u})}\mathbb{E}[(\hat{z}(\hat{U})-\hat{X})^2]$, where $\hat{X}\sim \hat{p}_0$ and $\mathcal{Q}=\{p(\hat{u}|\hat{x}): R\geq H(\hat{X}), R_s\geq H(\hat{X}|\hat{U})\}$. If we fix $(R, R_s)$ and suppose $p(\hat{u}|\hat{x})\in \mathcal{Q}$ is the corresponding distribution that achieves $D_{\hat{p}_0}(R,R_s)$, $\mathcal{T}_n\triangleq \{t_i(\hat{m},s^{E}_i)\}_{i=1}^n$, $\hat{\mathcal{T}}_n\triangleq \{\hat{t}_i(\hat{m},\hat{s}^{E}_i)\}_{i=1}^n$, then we have the following inequalities:
\begin{eqnarray}
&&\sigma_0^2 \Pi_{p_0}^{\Delta}(R,R_s) \nonumber\\
&\overset{(a)}=&\lim_{n\rightarrow \infty}\sup_{\mathcal{B}_n}\inf_{\mathcal{T}_n}\frac{1}{n}\sum_{i=1}^n\mathbb{E}[(Z_i-X_i)^2]-\mathbb{E}[(X_i-\hat{Y}_i)^2]\nonumber\\
&\overset{(b)}=&\lim_{n\rightarrow \infty}\sup_{\mathcal{B}_n}\inf_{\hat{\mathcal{T}}_n}\frac{1}{n}\sum_{i=1}^n\mathbb{E}[(\hat{Z}_i-X_i)^2]-\mathbb{E}[(X_i-\hat{Y}_i)^2]\nonumber\\
&\overset{(c)}=&\lim_{n\rightarrow \infty}\sup_{\mathcal{B}_n}\inf_{\hat{\mathcal{T}}_n}\frac{1}{n}\sum_{i=1}^n\mathbb{E}[(\hat{Z}_i-\hat{X}_i)^2]+\mathbb{E}[(\hat{X}_i-X_i)^2]\nonumber\\
&&+2\mathbb{E}\bigg{[}\mathbb{E}[(\hat{Z}_i-\hat{X}_i)(\hat{X}_i-X_i)|\hat{X}_i]\bigg{]}-\mathbb{E}[(X_i-\hat{Y}_i)^2]\nonumber\\
&\overset{(d)}=&\lim_{n\rightarrow \infty}\sup_{\mathcal{B}_n}\inf_{\hat{\mathcal{T}}_n}\frac{1}{n}\sum_{i=1}^n\mathbb{E}[(\hat{Z}_i-\hat{X}_i)^2]\nonumber\\
&&+\frac1n \sum_{i=1}^n\mathbb{E}[(\hat{X}_i-X_i)^2] -\mathbb{E}[(X_i-\hat{Y}_i)^2]\nonumber\\
&\overset{(e)}=&D_{\hat{p}_0}(R,R_s)
\end{eqnarray}
Here, (a) follows by definition of $\Pi^{\Delta}(R,R_s)$; (b) follows from Lemma \ref{lemma1}; (c) follows from law of total expectation; (d) follows from Lemma \ref{lemma2} since $\hat{Z}_i$ is a function of $(\hat{M},\hat{X}^{i-1})$; and (e) follows by Definition \ref{definition2}. Summarizing the analysis in this section, we have the following theorem.
\begin{thm} \label{theorem4}
$\Pi_{p_0}^{\Delta}(R,R_s) = \frac{1}{\sigma_0^2}D_{\hat{p}_0}(R,R_s)$.
\end{thm}
$D_{\hat{p}_0}(R,R_s)$ can be calculated as a linear program (LP). The technical details of this LP can be found in \cite{globecom-cuff}.

\begin{figure}[htbp]
  \centering
  \includegraphics[width=7.5 cm]{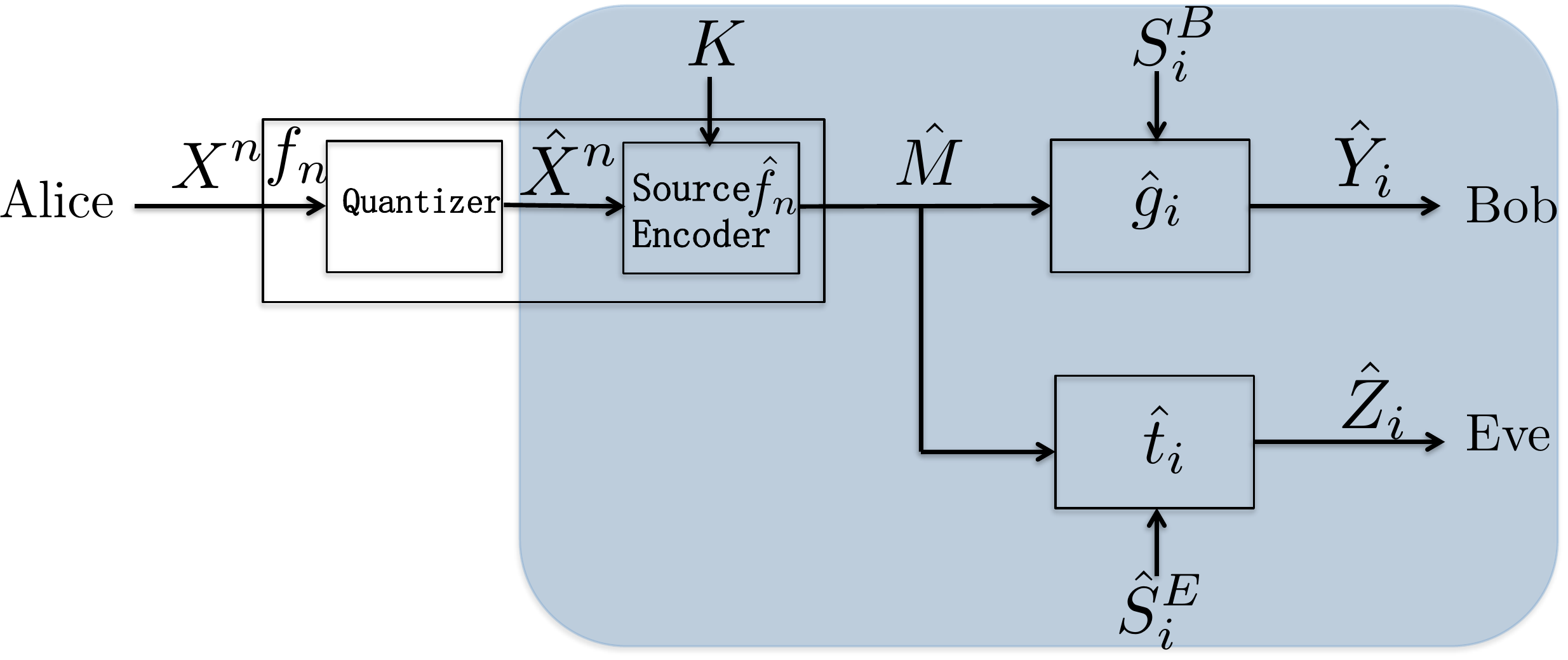}
\caption{An alternative model in which quantization is performed before lossless compression. }
\label{alternative}
\end{figure}

\section{Numerical Computation Result} \label{simulation}
In this section, we numerically compare the payoffs under different schemes for the scenario of causal source awareness. Since the optimization for the causal general awareness case involves two auxiliary random variables and displays more complexity, we numerically analyze only the causal source awareness case as an example. 

We compare the Gaussian quantization scheme with the jointly Gaussian scheme. Even though in Section \ref{sec-quan}, we gave only an analytical lower bound on the payoff as a function of the rates as $R\rightarrow \infty$, here we propose a numerical scheme that can evaluate the achievable secrecy rate-payoff for arbitary $R$ and $R_s$. The choice of the random variable $Y$ is the same as in Section \ref{sec-quan} and $U\triangleq n \mod N$, where $N$ is some positive integer. Intuitively, $U$ is a coarser quantizer of $X$. Here we greedily obtain an achievable lower bound by sequentially solving for the optimal $T$ that satisfies $R\geq I(X;U,Y)$ and the optimal $N$ that satisfies $R_s\geq I(X;Y|U)$. The payoff of the optimal scheme for $R_s\geq 1$ bit is also computed for low $R_s$ for comparison. These results are shown in Fig. \ref{quant-res}, which shows that the Gaussian quantization choice outperforms the jointly Gaussian choice in the payoff as a function of $R_s$ for a fixed $R$. The quantization upper bound is numerically obtained by solving the LP. Note that even though Theorem \ref{theorem4} gives a tight bound, the implementation of LP requires that the eavesdropper's reconstruction fall in the same quantization alphabet. 
%  \centering
%		 \includegraphics[width=9 cm]{numeric_quan.png}
% 		 \caption{Quantization of $X$ with interval $T$ and $N=4$}
%		 \label{quant-bell}
%\end{figure}
\begin{figure}[htbp]
	\centering
	\includegraphics[width=8.5 cm]{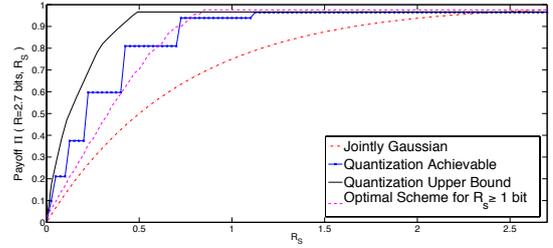}
	\caption{Payoff as a function of key rate $R_s$ for fixed $R=2.7 \text{ bits}$}
	\label{quant-res}
\end{figure}

\section{Conclusion}
We have applied the secrecy source coding problem introduced in \cite{allerton-cuff} for the important case of Gaussian sources. The payoff function has been defined in a way such that the reconstruction qualities (measured by mean squared error) of the source from Bob and from Eve are compared. Three scenarios depending on the SI to Eve (from weak to strong) have been investigated. Our results show 1) for a weak Eve who has no causal information about the source, maximum distortion between Alice and Eve can be achieved with no loss of quality in the communication between Alice and Bob with any positive key rate; and 2) for an Eve that has causal information of Alice (and Bob), at most $1$ bit of secret key is needed for each Gaussian source symbol to force maximum distortion to Eve while keeping the rate-distortion tradeoff at the same level as in point-to-point communication. 
%It is not surprising that if we hide the sign of a Gaussian source symbol, Eve cannot make any sense out of the higher precision bits. When the key rate is less than $1$ bit (and $R_s<R$), we provided a quantization scheme that numerically defeated the choice of jointly Gaussian auxiliary random variable and gave an upper bound on the payoff for the quantization case.
This implies that even if a very high rate is demanded by Bob to have a good quality reconstruction of the source, no more secret key is needed, because the ``higher precision" bits do not leak out any
information helpful to Eve. When the key rate is less than $1$ bit (and $R_s < R$), the optimization problem cannot be decoupled and the constraint region is rather complicated. Counter-intuitively, choosing the auxiliary random variable
to be jointly Gaussian is not optimal. This problem is still left open from this work. However, the example given in the numerical results section indicates there will be a tradeoff between secrecy and reliability.
%% Appendix:
%% If needed a single appendix is created by
%\appendix
%% If several appendices are needed, then the command
%\appendices
%% in combination with further \section-commands can be used.

%% Use \section* for acknowledgement
\section*{Acknowledgment}
This work was supported by NSF Grants CCF-1116013 and CNS-09-05086, and AFOSR Grant FA9550-12-1-0196. The authors would like to thank Curt Schieler for providing the code for linear programming.

%% References:
%% We recommend the usage of BibTeX:
%%
%\bibliographystyle{IEEEtran}
%\bibliography{definitions,bibliofile}
%%
%% where we here have assume the existence of the files
%% definitions.bib and bibliofile.bib.
%% BibTeX documentation can be obtained at:
%% http://www.ctan.org/tex-archive/biblio/bibtex/contrib/doc/
%%
%%
%%
%% Or manual references (pay attention to consistency!):

\end{document}